\documentclass[journal]{IEEEtran}

\usepackage[cmex10]{amsmath}
\usepackage{amssymb,amsfonts,amsthm}
\usepackage{booktabs}
\usepackage{graphicx}
\usepackage{cite}
\usepackage[colorlinks=false,pdfborder={0 0 0}]{hyperref}
\newcommand{\customfootnotetext}[2]{{
  \renewcommand{\thefootnote}{#1}
  \footnotetext[0]{#2}}}
\newtheorem{theorem}{Theorem}
\newtheorem{lemma}[theorem]{Lemma}
\newtheorem{proposition}[theorem]{Proposition}
\newtheorem{corollary}[theorem]{Corollary}

\newtheorem{remark}[theorem]{Remark}

\newcommand{\ket}[1]{\lvert #1\rangle}
\newcommand{\bra}[1]{\langle #1\rvert}

\begin{document}

\title{Haar-Bayesian Pure-State Prediction under Relative-Entropy Loss: Arbitrary-Effect Reduction and Global Optimality}
\author{Masahito Hayashi\textsuperscript{$*$}, 
Ayanava Dasgupta\textsuperscript{$\dagger$} and Naqueeb Ahmad Warsi\textsuperscript{$\dagger$}
}
\date{}

\maketitle
\customfootnotetext{$*$}{
School of Data Science, The Chinese University of Hong Kong, Shenzhen, Longgang District, Shenzhen, 518172, China,
International Quantum Academy (IQA), Futian District, Shenzhen 518048, China,
and
Graduate School of Mathematics, Nagoya University, Chikusa-ku, Nagoya 464-8602, Japan.
(e-mail: hmasahito@cuhk.edu.cn).
}

\customfootnotetext{$\dagger$}{
Indian Statistical Institute,
Kolkata 700108, India.
(e-mail: [ayanavadasgupta\_r, naqueebwarsi]@isical.ac.in
)}

\begin{abstract}
We study Haar-Bayesian prediction of one unmeasured copy of an unknown
finite-dimensional pure quantum state after an arbitrary collective measurement
on \(n\) observed copies.  Performance is evaluated by quantum relative entropy.
For a fixed measurement, the Bayes predictive state is the posterior mean and
the optimized conditional loss is its entropy.  We then optimize the measurement
over all POVMs on the symmetric subspace.  For every nonzero positive effect
\(E\), the corresponding posterior predictive state is
\(\mu_E=(I+n\rho_E)/(n+d)\), where \(\rho_E\) is the normalized one-particle
marginal of \(E\).  Since a pure spectrum majorizes every density-operator
spectrum, this identity gives an outcome-wise entropy lower bound.  Coherent
rank-one effects attain the bound, and their Haar orbit yields the
highest-weight covariant POVM.  Hence this POVM is globally Bayes optimal over
all collective measurements and, by covariance, globally minimax.  Its exact risk is \(h_d((n+1)/(n+d))\), where
\(h_d(r)=-r\log r-(1-r)\log((1-r)/(d-1))\).
The same arbitrary-effect reduction shows that the highest-weight POVM also
maximizes the joint overlap between the latent pure state and its posterior
predictive state, equivalently the mean posterior purity, with optimum
\(((n+1)^2+d-1)/(n+d)^2\).
\end{abstract}

\section{Introduction}
In quantum inference, measurements on some identically prepared systems can be
used not only to estimate their common state, but also to assign a state to
another system that remains unmeasured
\cite{SchackBrunCaves2001QuantumBayes,TanakaKomaki2005BayesianPredictive}.
Among several systems prepared in the same unknown state, some copies may be
measured for learning, while another copy is retained for subsequent use.  The
relevant task is then not merely to report an estimate of the common state, but
to determine what state should be assigned to the unmeasured system in light of
the observed data.  This leads naturally to a quantum prediction problem.  The
prediction is a density operator conditioned on the measurement outcomes,
rather than necessarily a single estimated pure state, and can therefore
incorporate both the information obtained from the measured copies and the
uncertainty that remains after the measurement.

In this paper, we study this problem for finite-dimensional pure states under a
uniform prior.  An unknown pure state is drawn from the Haar distribution, the
first \(n\) copies are measured, and one additional copy prepared in the same
state is left unmeasured.  On the basis of the measurement outcomes, a density
operator is assigned to the remaining copy, and its quality is evaluated by
quantum relative-entropy loss.  The main questions are how the predictive state
should be chosen for a given measurement and which collective measurement gives
the best prediction.

The corresponding classical theory is predictive density estimation under
Kullback--Leibler loss: the announced predictive density is assessed by its
divergence from the true future density, and the posterior predictive density
is the Bayes act.  This viewpoint goes back to Aitchison's prediction-fit
formulation and to the predictive analyses of Aitchison--Dunsmore and Geisser,
and it has been developed further in Bayesian predictive-density estimation
under KL risk
\cite{Aitchison1975PredictionFit,AitchisonDunsmore1975Prediction,Geisser1993Predictive,Komaki2006ShrinkagePrediction,GeorgeLiangXu2006PredictiveKL}.
Generalized Bayes rules for prediction under other divergences provide a related
classical extension \cite{CorcueraGiummole1999Prediction}.

In the quantum setting, Tanaka and Komaki formulated Bayesian predictive density
operators and proved their optimality under averaged quantum relative entropy
\cite{TanakaKomaki2005BayesianPredictive}.  Relatedly, the quantum Bayes rule for
exchangeable states explains how posterior states arise after measurements
\cite{SchackBrunCaves2001QuantumBayes}.  Tanaka's generalized predictive density
operators extend this direction to quantum divergence families
\cite{Tanaka2006GeneralizedPredictive}.  Accordingly, the fixed-measurement
principle used below---that the posterior predictive state is the Bayes act---is
a known general foundation.  We include a self-contained proof specialized to
the present Haar-prior pure-state model because it is the starting point for the
subsequent optimization over measurements.  Likewise, compact-group averaging,
covariant decision rules, and the resulting Bayes--minimax comparison are used
as standard decision-theoretic tools
\cite{Holevo1979CovariantMeasurements,Holevo1982ProbabilisticStatistical}.
The global Bayes optimality theorem itself, however, is proved directly for
arbitrary POVM effects and does not rely on restricting the optimization to a
particular covariant outcome space.

Relative entropy makes the support of the predictive state important: a
rank-deficient predictive state can incur infinite loss against a pure true
state.  The arbitrary-effect posterior identity contains a universal identity
term.  This term ensures full rank and provides the spectral floor needed to
keep the relative-entropy loss finite.

The main contribution is a direct global optimization over arbitrary collective
POVMs on the symmetric subspace of the observed copies.  For any nonzero
positive effect \(E\) on \(\operatorname{Sym}^n(\mathcal H)\), we derive the
effect-conditioned posterior predictive state
\[
    \mu_E=\frac{I+n\rho_E}{n+d},
\]
where \(\rho_E\) is the normalized one-particle marginal of \(E\).  Thus the
optimized conditional loss depends on the effect only through the spectrum of
a single one-copy density operator.  Since the spectrum of a pure state
majorizes that of every density operator, Schur concavity gives a common entropy
lower bound at every measurement outcome.  Coherent rank-one effects attain
this bound outcome by outcome, and their Haar orbit forms the highest-weight
covariant POVM.  Averaging the pointwise bound therefore proves global Bayes
optimality over all collective POVMs; covariance then yields global minimax
optimality.  Coherent or highest-weight measurements are familiar from
fidelity-based pure-state estimation
\cite{MassarPopescu1995OptimalEstimation,DerkaBuzekEkert1998OptimalEstimation,Hayashi1998PureStateEstimation,Group2}.
The present issue is different: the decision is a generally mixed predictive
state, the loss is quantum relative entropy, and optimality is obtained by a
direct arbitrary-effect entropy argument.  To the best of our knowledge, this
arbitrary-effect posterior identity and the resulting outcome-wise
majorization proof of global relative-entropy optimality have not previously
been stated in this form.  The same identity also shows that the highest-weight
measurement maximizes the joint posterior-overlap criterion
\(\mathbb E_{\Phi,Y}\operatorname{Tr}[\rho_\Phi\rho_M(Y)]\), equivalently the mean
purity of the posterior predictive state.  Thus the global optimizer is shared
by relative-entropy prediction and this quadratic posterior criterion, although
their objective functions are different.

A companion analysis fixes the observation itself by applying a one-copy Haar
POVM independently to each observed copy.  That problem is distinct from the
collective optimization treated here: it evaluates the exact posterior for the
fixed product observation and studies its finite-sample and fixed-dimensional
asymptotic performance relative to the global benchmarks established below.

The remainder of the paper is organized as follows.
Section~\ref{sec:general_decision} establishes the posterior-mean Bayes rule and
the probability notation used to distinguish input-fixed, output-fixed, and
joint averages.  Section~\ref{sec:covariant_reduction} develops the covariance,
Bayes--minimax, depolarizing, and coarse-graining tools.  Section~\ref{sec:global_optimization}
proves the arbitrary-effect posterior identity, global Bayes and minimax
optimality, and the global posterior-overlap optimum.  The final section
discusses equality conditions, scope, and extensions.
\section{Bayesian Quantum Prediction under Relative-Entropy Loss}
\label{sec:general_decision}
This section formulates the measurement-independent prediction problem used
throughout the paper.  An unknown pure state \(\rho_\phi\) is drawn from the Haar
prior, and \(n+1\) copies \(\rho_\phi^{\otimes(n+1)}\) are regarded as
exchangeable systems.  A measurement is performed only on the first \(n\) copies,
producing an outcome \(x\).  The last copy is left unmeasured and is the system
to be predicted.  After observing \(x\), a predictive rule announces a density
matrix \(\sigma(x)\) for this \((n+1)\)-st copy.  The loss for the true latent
state \(\phi\) is the quantum relative entropy \(D(\rho_\phi\Vert\sigma(x))\).
The posterior distribution of \(\phi\) given the measurement outcome \(x\) is
therefore the relevant object for prediction: the Bayes predictive state is
obtained by averaging the one-copy state \(\rho_\phi\) with respect to this
posterior.

Throughout the paper, the number of observed copies is denoted by \(n\),
and the predictive target is one additional unmeasured copy prepared in the
same unknown state.

\begin{table*}[t]
\centering
\caption{Basic notation for the quantum prediction problem.}
\label{tab:basic_prediction_notation}
\small
\renewcommand{\arraystretch}{1.08}
\begin{tabular}{p{0.15\textwidth}p{0.27\textwidth}p{0.48\textwidth}}
\toprule
Symbol & Object & Meaning or defining relation \\
\midrule
\(d\) & Hilbert-space dimension & \(\mathcal H\simeq\mathbb C^d\). \\
\(n\) & Observed-copy count & Number of copies measured before predicting one additional unmeasured copy. \\
\(\rho_\phi\) & Unknown pure state & \(\rho_\phi=\ket\phi\bra\phi\), with Haar prior \(\nu\). \\
\(M^{(n)}(dx)\) & Measurement & General POVM on the observed copies, with outcome \(x\in\mathsf X\). \\
\(\pi_M(d\phi\mid x)\) & Posterior distribution & Conditional distribution of the unknown ray after observing \(x\). \\
\(\sigma(x)\) & Announced predictive state & Density operator assigned to the additional unmeasured copy. \\
\(\mu_M(x)\) & Posterior predictive state & \(\mu_M(x)=\int\rho_\phi\,\pi_M(d\phi\mid x)\); the fixed-measurement Bayes act. \\
\(D(\rho\Vert\sigma)\) & Loss & Quantum relative entropy from the true state to the announced state. \\
\(S(\rho)\) & Entropy & The minimized posterior loss is \(S(\mu_M(x))\). \\
\(R(M,\sigma)\), \(R(M)\) & Bayes risks & Risk before and after optimization over the predictive rule, respectively. \\
\bottomrule
\end{tabular}
\end{table*}

\subsection{Observed copies and the future copy}
\label{subsec:general_model}
Let \(\mathcal H\simeq\mathbb C^d\).  The unknown one-copy state is a pure state
\[
    \rho_\phi:=\ket{\phi}\bra{\phi},
\]
where \(\phi\) denotes a ray in \(\mathbb P(\mathcal H)\), equipped with the Haar
probability measure \(\nu(d\phi)\).  For the general fixed-measurement problem,
let \(M^{(n)}(dx)\) be a POVM on the first \(n\) copies, with outcome space
\(\mathsf X\).  The \((n+1)\)-st copy is not measured by \(M^{(n)}\); it is the
future system to be predicted.

For a true state \(\rho_\phi\), the outcome distribution is
\[
    p_M(dx\mid\phi)
    :=
    \operatorname{Tr}\!\left[\rho_\phi^{\otimes n}M^{(n)}(dx)\right].
\]
The prior predictive outcome distribution is
\[
    p_M(dx)
    :=
    \int_{\mathbb P(\mathcal H)}p_M(dx\mid\phi)\,\nu(d\phi).
\]
For \(p_M\)-almost every outcome \(x\), let \(\pi_M(d\phi\mid x)\) denote the
posterior distribution induced by Bayes' rule.  A predictive rule, equivalently
a density-matrix-valued decision rule, is a measurable map
\[
    x\longmapsto \sigma(x)\in\mathcal S(\mathcal H),
\]
where \(\mathcal S(\mathcal H)\) is the set of density matrices on the one-copy
Hilbert space.  The announced state \(\sigma(x)\) is interpreted as the
prediction assigned to the unmeasured \((n+1)\)-st copy.  We use the quantum
relative entropy~\cite{umegaki1962}
\[
    D(\rho\Vert\sigma),
\]
with the standard convention that it is \(+\infty\) when the support of \(\rho\)
is not contained in the support of \(\sigma\).  The Bayes prediction risk of
\((M^{(n)},\sigma)\) is
\[
    R(M,\sigma)
    :=
    \int_{\mathbb P(\mathcal H)}\int_{\mathsf X}
    D(\rho_\phi\Vert\sigma(x))\,
    p_M(dx\mid\phi)\,\nu(d\phi).
\]
Equivalently, by disintegration,
\[
    R(M,\sigma)
    =
    \int_{\mathsf X}
    \left[
        \int D(\rho_\phi\Vert\sigma(x))\,
        \pi_M(d\phi\mid x)
    \right]
    p_M(dx).
\]

\subsection{Input-fixed, output-fixed, and joint averages}
\label{subsec:input_output_joint_averages}

Before optimizing the predictive state, we record a common probability
structure that will also be used in the global relative-entropy and posterior-overlap comparisons.  This notation distinguishes the physical input state, the
observed data, and the state eventually announced from those data.

Let \(\mathcal P=(M^{(n)},\sigma)\) denote a procedure consisting of a
measurement \(M^{(n)}\) and a density-matrix-valued decision rule \(\sigma\).
Its observation is denoted by \(Y_{\mathcal P}\); in the present general model
\(Y_{\mathcal P}=x\in\mathsf X\).  The observation is not itself the announced
state: the latter is \(\sigma(Y_{\mathcal P})\).  For a fixed input ray \(\phi\),
the conditional law of the observation is
\begin{equation}
    Q_{\mathcal P}(dy\mid\phi)
    :=
    p_M(dy\mid\phi).
    \label{eq:procedure_conditional_output_law}
\end{equation}
The Haar prior and this conditional law define the joint distribution
\begin{equation}
    J_{\mathcal P}(d\phi,dy)
    :=
    \nu(d\phi)Q_{\mathcal P}(dy\mid\phi).
    \label{eq:procedure_joint_input_output_law}
\end{equation}
Its observation marginal is
\begin{equation}
    q_{\mathcal P}(dy)
    :=
    \int Q_{\mathcal P}(dy\mid\phi)\,\nu(d\phi),
    \label{eq:procedure_output_marginal_law}
\end{equation}
and its conditional input distribution given the observation is the posterior
\(\pi_{\mathcal P}(d\phi\mid y)\).  Thus, up to the usual null sets,
\begin{equation}
    J_{\mathcal P}(d\phi,dy)
    =
    q_{\mathcal P}(dy)\pi_{\mathcal P}(d\phi\mid y).
    \label{eq:procedure_joint_law_disintegration}
\end{equation}
For the fixed measurement notation used above,
\(q_{\mathcal P}=p_M\) and \(\pi_{\mathcal P}=\pi_M\).

Let \(h_{\mathcal P}(\phi,y)\) be any real-valued function integrable with
respect to \(J_{\mathcal P}\).  We distinguish three associated quantities.
The input-fixed value is
\begin{equation}
    H_{\mathcal P}^{\mathrm{in}}(\phi)
    :=
    \int h_{\mathcal P}(\phi,y)Q_{\mathcal P}(dy\mid\phi),
    \label{eq:procedure_input_fixed_value}
\end{equation}
the output-fixed posterior value is
\begin{equation}
    H_{\mathcal P}^{\mathrm{out}}(y)
    :=
    \int h_{\mathcal P}(\phi,y)\pi_{\mathcal P}(d\phi\mid y),
    \label{eq:procedure_output_fixed_value}
\end{equation}
and the joint average is
\begin{equation}
    H_{\mathcal P}^{\mathrm{joint}}
    :=
    \iint h_{\mathcal P}(\phi,y)J_{\mathcal P}(d\phi,dy).
    \label{eq:procedure_joint_average_value}
\end{equation}
Here ``input-fixed'' means that the true ray is fixed and only the observation
is averaged, whereas ``output-fixed'' means that the observation is fixed and
the input ray is averaged under its posterior distribution.

\begin{lemma}[Input-first and output-first averaging]
\label{lem:input_output_average_identity}
For every \(J_{\mathcal P}\)-integrable function
\(h_{\mathcal P}\),
\begin{align}
    H_{\mathcal P}^{\mathrm{joint}}
    &=
    \int H_{\mathcal P}^{\mathrm{in}}(\phi)\,\nu(d\phi)
    \notag\\
    &=
    \int H_{\mathcal P}^{\mathrm{out}}(y)\,q_{\mathcal P}(dy).
    \label{eq:input_output_average_identity}
\end{align}
Equivalently,
\begin{align}
&\int \nu(d\phi)
    \int Q_{\mathcal P}(dy\mid\phi)
    h_{\mathcal P}(\phi,y)
\notag\\
&\qquad=
\int q_{\mathcal P}(dy)
    \int \pi_{\mathcal P}(d\phi\mid y)
    h_{\mathcal P}(\phi,y).
    \label{eq:input_first_output_first_expanded}
\end{align}
Thus, after both the input and the observation have been averaged, the value is
independent of whether one first fixes the input or first fixes the observation.
\end{lemma}

\begin{proof}
Both iterated integrals are integrals of the same function with respect to the
same joint law \(J_{\mathcal P}\).  The first equality uses
\eqref{eq:procedure_joint_input_output_law}; the second uses the disintegration
\eqref{eq:procedure_joint_law_disintegration}.  Fubini's theorem gives the
claim.
\end{proof}

\begin{corollary}[Constant conditional values]
\label{cor:constant_conditional_values_joint_average}
If \(H_{\mathcal P}^{\mathrm{in}}(\phi)\) is independent of \(\phi\), then its
common value equals \(H_{\mathcal P}^{\mathrm{joint}}\).  If
\(H_{\mathcal P}^{\mathrm{out}}(y)\) is independent of \(y\), then its common
value also equals \(H_{\mathcal P}^{\mathrm{joint}}\).  Hence, whenever both
conditional values are constant,
\begin{equation}
    H_{\mathcal P}^{\mathrm{in}}
    =
    H_{\mathcal P}^{\mathrm{joint}}
    =
    H_{\mathcal P}^{\mathrm{out}}.
    \label{eq:constant_input_output_joint_values}
\end{equation}
\end{corollary}

\begin{proof}
This follows immediately from
\eqref{eq:input_output_average_identity}, because the Haar prior and the
observation marginal are probability measures.
\end{proof}

The constancy required in the corollary is not automatic.  In later
applications it follows from unitary covariance and transitivity on the
relevant input or output orbit.  For a general density-matrix-valued output
space, different output spectra may belong to different unitary orbits, so an
output-fixed quantity need not be constant without an additional argument.

Two special choices of \(h_{\mathcal P}\) will be used repeatedly.  The
relative-entropy loss is
\begin{equation}
    h_{\mathcal P}^{\mathrm{RE}}(\phi,y)
    :=
    D\!\left(\rho_\phi\middle\|\sigma(y)\right),
    \label{eq:procedure_relative_entropy_performance_function}
\end{equation}
and the fidelity between the pure input and the announced state is
\begin{equation}
    h_{\mathcal P}^{\mathrm{fid}}(\phi,y)
    :=
    \operatorname{Tr}\!\left[\rho_\phi\sigma(y)\right].
    \label{eq:procedure_fidelity_performance_function}
\end{equation}
For a pure-state-valued decision rule
\(\sigma(y)=\ket{\widehat\psi(y)}\bra{\widehat\psi(y)}\), the latter becomes
\begin{equation}
    h_{\mathcal P}^{\mathrm{fid}}(\phi,y)
    =
    \left|\langle\widehat\psi(y)\mid\phi\rangle\right|^2.
    \label{eq:procedure_fidelity_pure_output_form}
\end{equation}
Accordingly, the same input-fixed, output-fixed, and joint-average conventions
apply to the global relative-entropy and posterior-overlap criteria.
Procedure-specific abbreviations will be introduced only after the relevant
equality of conditional and joint values has been established.

\subsection{Posterior predictive state and posterior entropy}
\label{subsec:posterior_predictive_state}
\label{subsec:bayes_act_fixed_measurement}
\label{subsec:posterior_entropy_risk}
For a fixed measurement \(M^{(n)}\) on the observed copies, define the posterior
predictive state
\[
    \mu_M(x)
    :=
    \int \rho_\phi\,\pi_M(d\phi\mid x).
\]
This is a one-copy density matrix.  It is the posterior mean of the latent
one-copy state and is announced as the predictive state for the unmeasured
\((n+1)\)-st copy.  The Bayesian predictive-state principle in the theorem
below is the specialization to the present Haar-prior pure-state model of the
general predictive-density-operator result under averaged quantum relative
entropy \cite{TanakaKomaki2005BayesianPredictive}.  We include the proof to keep
the subsequent optimization over measurements self-contained and to record the
posterior entropy that is optimized in Section~\ref{sec:global_optimization}.

\begin{theorem}[Posterior predictive state and posterior entropy]
\label{thm:fixed_measurement_bayes_act}
\label{cor:fixed_measurement_entropy_risk}
For any fixed POVM \(M^{(n)}\) on the observed copies, and for \(p_M\)-almost
every outcome \(x\), the posterior expected relative-entropy loss
\[
    \sigma\longmapsto
    \int D(\rho_\phi\Vert\sigma)\,\pi_M(d\phi\mid x)
\]
over density matrices \(\sigma\in\mathcal S(\mathcal H)\) is minimized by
\(\sigma=\mu_M(x)\).  Moreover, the minimum value is \(S(\mu_M(x))\).  Hence the
optimized fixed-measurement prediction risk is
\[
    R(M)
    :=
    \inf_\sigma R(M,\sigma)
    =
    \int_{\mathsf X} S(\mu_M(x))\,p_M(dx).
\]
\end{theorem}

\begin{proof}
Fix an outcome \(x\) for which the posterior is defined, and write
\[
    L_x(\sigma)
    :=
    \int D(\rho_\phi\Vert\sigma)\,\pi_M(d\phi\mid x).
\]
If \(\operatorname{supp}\mu_M(x)\nsubseteq\operatorname{supp}\sigma\), let
\(P_0\) be the projector onto \(\ker\sigma\).  Then
\(\operatorname{Tr}[\mu_M(x)P_0]>0\).  Since
\[
    \operatorname{Tr}[\mu_M(x)P_0]
    =
    \int \operatorname{Tr}[\rho_\phi P_0]\,\pi_M(d\phi\mid x),
\]
the set of rays for which \(\operatorname{Tr}[\rho_\phi P_0]>0\) has positive
posterior measure.  On that set,
\(\operatorname{supp}\rho_\phi\nsubseteq\operatorname{supp}\sigma\), so the
relative-entropy loss is infinite.  Hence \(L_x(\sigma)=+\infty\), and a finite
minimizer must satisfy
\[
    \operatorname{supp}\mu_M(x)
    \subseteq
    \operatorname{supp}\sigma.
\]
For such \(\sigma\), since each \(\rho_\phi\) is pure,
\[
    D(\rho_\phi\Vert\sigma)
    =
    -\operatorname{Tr}\rho_\phi\log\sigma.
\]
Therefore
\[
\begin{aligned}
    L_x(\sigma)
    &=
    -\int \operatorname{Tr}\rho_\phi\log\sigma\,\pi_M(d\phi\mid x)       \\
    &=
    -\operatorname{Tr}\!\left[
        \left(\int \rho_\phi\,\pi_M(d\phi\mid x)\right)
        \log\sigma
    \right]                                                          \\
    &=
    -\operatorname{Tr}\mu_M(x)\log\sigma                              \\
    &=
    S(\mu_M(x))+D(\mu_M(x)\Vert\sigma).
\end{aligned}
\]
The first term depends only on the posterior generated by the measurement and
the observed outcome, while the second term is nonnegative and vanishes exactly
when \(\sigma=\mu_M(x)\).  Thus the posterior predictive state is
\(\mu_M(x)\), and the minimized posterior loss is \(S(\mu_M(x))\).  Integrating
this pointwise minimum with respect to the prior predictive distribution
\(p_M(dx)\) gives the displayed formula for \(R(M)\).
\end{proof}

Thus, once the measurement on the observed copies is fixed, the only remaining
measurement-dependent quantity is the average entropy of the posterior
predictive state.  This is the quantity optimized over all collective measurements in
Section~\ref{sec:global_optimization}.

\subsection{Pure-state geometry}
\label{subsec:pure_state_fidelity_geometry}
We use the following pure-state fidelity notation throughout the paper.  For
rays \(\phi,\psi\in\mathbb P(\mathcal H)\), let
\begin{equation}
    f(\phi,\psi)
    :=
    \operatorname{Tr}[\rho_\phi\rho_\psi]
    =
    |\langle\phi\mid\psi\rangle|^2.
    \label{eq:pure_state_fidelity_definition}
\end{equation}
The squared pure-state infidelity, denoted by \(d_b^2\), is
\begin{equation}
    d_b^2(\rho_\phi,\rho_\psi)
    :=
    1-f(\phi,\psi)
    =
    1-|\langle\phi\mid\psi\rangle|^2.
    \label{eq:db_squared_pure_state_infidelity}
\end{equation}
We also use the Fubini--Study distance on projective space normalized by
\begin{equation}
    d_{\mathrm{FS}}([\phi],[\psi])
    :=
    \arccos |\langle\phi\mid\psi\rangle|.
    \label{eq:FS_distance_definition_global}
\end{equation}
With this convention,
\begin{equation}
    f(\phi,\psi)
    =
    \cos^2 d_{\mathrm{FS}}([\phi],[\psi]),
    ~
    d_b^2(\rho_\phi,\rho_\psi)
    =
    \sin^2 d_{\mathrm{FS}}([\phi],[\psi]).
    \label{eq:infidelity_FS_relation_global}
\end{equation}
Thus pure-state infidelity and the Fubini--Study distance are related by
this fixed convention.

\subsection{Randomized measurements and affinity of the prediction risk}
\label{subsec:randomized_measurements_affinity}
Let \(M_1(dx_1)\) and \(M_2(dx_2)\) be two POVMs on the observed copies.  Their
branch-labelled randomized mixture with weights \(0\le\lambda\le1\) is the POVM
\[
    M:=\lambda M_1\oplus(1-\lambda)M_2
\]
with outcome space
\[
    (\{1\}\times\mathsf X_1)\sqcup(\{2\}\times\mathsf X_2).
\]
The branch label is part of the observed data.  Conditioned on branch \(j\), the
posterior distribution and posterior predictive state are exactly those
generated by \(M_j\).
\begin{lemma}[Affinity under branch-labelled randomized measurements]
\label{lem:measurement_affinity}
For branch-labelled randomized measurements,
\[
    R(\lambda M_1\oplus(1-\lambda)M_2)
    =
    \lambda R(M_1)+(1-\lambda)R(M_2).
\]
\end{lemma}
\begin{proof}
The prior predictive outcome distribution of the randomized experiment is the
corresponding branch-labelled mixture of the two outcome distributions.  On
branch \(j\), the posterior predictive state and its entropy coincide with those
of \(M_j\).  Integrating over the two branches gives the displayed equality.
\end{proof}
This affine behavior will also be useful whenever a measurement is explicitly
implemented as a branch-labelled randomization.  The global proof below,
however, does not require a seed or orbit decomposition.

\section{Covariance Tools for Bayesian Prediction}
\label{sec:covariant_reduction}

This section records covariance tools used for Bayes--minimax comparisons and
for interpreting posterior states under equivariant coarse-graining.  The only
part needed for the main global minimax conclusion is the constant-risk
consequence of covariance.  In particular, the global Bayes optimum in
Section~\ref{sec:global_optimization} is proved directly for arbitrary POVM
effects and does not rely on the Hunt--Stein reduction or on a restriction to
projective-space outcomes.

\subsection{Haar prior and unitary covariance}
\label{subsec:haar_covariance}

The unitary group \(U(d)\) acts transitively on pure states by
\[
    \rho_\phi\longmapsto U\rho_\phi U^\dagger.
\]
On the \(n\)-copy input space the induced action is \(U^{\otimes n}\),
restricted to \(\operatorname{Sym}^n(\mathcal H)\).  The Haar prior is
invariant under this action, and relative entropy satisfies
\[
    D(U\rho U^\dagger\Vert U\sigma U^\dagger)
    =
    D(\rho\Vert\sigma).
\]
Thus the full Bayes problem is covariant.  In the present manuscript,
covariance is used for Bayes--minimax comparisons and for the expected-state
depolarizing arguments below.  The global Bayes optimum itself is proved
directly in Section~\ref{sec:global_optimization} by an outcome-wise
arbitrary-effect argument.

\subsection{Reduction to covariant decision procedures}
\label{subsec:covariant_povm_reduction}

We now prove the covariant reduction in the form needed in this paper.  The
statement is a direct Hunt--Stein averaging argument, adapted to the present
Bayesian relative-entropy decision problem.

Let
\[
    \Theta:=\mathbb P(\mathcal H)
\]
be the pure-state parameter space, equipped with its invariant probability
measure \(\mu_\Theta\).  Let
\[
    \rho_\theta:=\ket{\theta}\bra{\theta}
\]
and let the \(n\)-copy input state be \(\rho_\theta^{\otimes n}\).  The group
\(G=U(d)\) acts transitively on \(\Theta\), and we write
\[
    g\theta:=U_g\theta,
    \qquad
    \rho_{g\theta}=U_g\rho_\theta U_g^\dagger.
\]
On the input system the representation is \(U_g^{\otimes n}\).  The decision
space is the compact convex set
\[
    \mathcal A:=\mathcal S(\mathcal H)
\]
of density matrices on the single-copy Hilbert space, with group action
\[
    g\sigma:=U_g\sigma U_g^\dagger.
\]
The loss is
\[
    \ell(\theta,\sigma)
    :=
    D(\rho_\theta\Vert\sigma).
\]
It is invariant in the sense that
\[
    \ell(g\theta,g\sigma)=\ell(\theta,
    \sigma).
\]

A decision procedure can be represented as a POVM \(N(d\sigma)\) on the
decision space \(\mathcal A\): after observing the outcome \(\sigma\), the
procedure outputs that density matrix.  This representation includes the usual
``measurement plus decision rule'' formulation, because a POVM \(M(dx)\) and a
measurable decision rule \(x\mapsto\sigma(x)\) induce the push-forward POVM
\[
    N(B):=M\bigl(\{x:\sigma(x)\in B\}\bigr),
    \qquad B\subset\mathcal A.
\]
Conversely, a POVM on \(\mathcal A\) is simply a randomized decision rule whose
outcome is already the announced density matrix.

For a decision POVM \(N\), define its frequentist risk at \(\theta\) by
\[
    R_\theta(N)
    :=
    \int_{\mathcal A}
    \ell(\theta,\sigma)
    \operatorname{Tr}\!\left[
        \rho_\theta^{\otimes n}N(d\sigma)
    \right].
\]
The Haar-Bayes risk and the worst-case risk are
\[
    R_{\mu_\Theta}(N)
    :=
    \int_\Theta R_\theta(N)\,\mu_\Theta(d\theta),
    \quad
    R_{\max}(N):=\sup_{\theta\in\Theta}R_\theta(N).
\]
A decision POVM is called covariant if
\[
    N(gB)
    =
    U_g^{\otimes n}N(B)U_g^{\otimes n\dagger}
\]
for every Borel set \(B\subset\mathcal A\), where
\(gB:=\{g\sigma:\sigma\in B\}\).

The next theorem is the standard compact-group Hunt--Stein averaging argument
in the decision-theoretic form needed here
\cite{Holevo1979CovariantMeasurements,Holevo1982ProbabilisticStatistical}.
Its role is to support Bayes--minimax comparisons and constant-risk statements.
It is not used to restrict the outcome space in the proof of the global Bayes
optimum: Section~\ref{sec:global_optimization} proves that result directly for
the positive density of an arbitrary POVM.

\begin{theorem}[Hunt--Stein reduction for the present Bayes problem]
\label{thm:covariant_reduction}
For the Haar-prior pure-state model with relative-entropy loss, the minimum
Haar-Bayes risk over all decision procedures is equal to the minimum Haar-Bayes
risk over covariant decision procedures.  Moreover, for every decision POVM
\(N\), there exists a covariant decision POVM \(\overline N\) such that
\[
    R_{\mu_\Theta}(\overline N)=R_{\mu_\Theta}(N).
\]
For a covariant decision POVM \(\overline N\), the risk is constant on
\(\Theta\), and hence
\[
    R_{\max}(\overline N)=R_{\mu_\Theta}(\overline N).
\]
In particular, if a covariant procedure \(N_\star\) attains the Bayes optimum
\(B\), then its constant frequentist risk equals \(B\), and
\[
    B
    \le
    \inf_N R_{\max}(N)
    \le
    R_{\max}(N_\star)
    =
    R_{\mu_\Theta}(N_\star)
    =B.
\]
Hence \(N_\star\) is also minimax and the minimax value equals \(B\).
\end{theorem}

\begin{proof}
Fix an
arbitrary decision POVM \(N\).  For each \(g\in G\), define the rotated decision
POVM
\[
    N_g(B)
    :=
    U_g^{\otimes n\dagger}N(gB)U_g^{\otimes n},
    \qquad B\subset\mathcal A.
    \label{eq:rotated_decision_povm_def}
\]
We first compute its risk.  Using the change of decision variable
\(\omega=g\sigma\), the covariance of the state family, and the
loss-invariance relation displayed above, we obtain
\[
\begin{aligned}
    R_\theta(N_g)
    &=
    \int_{\mathcal A}
    \ell(\theta,\sigma)
    \operatorname{Tr}\!\left[
        \rho_\theta^{\otimes n}
        U_g^{\otimes n\dagger}N(gd\sigma)U_g^{\otimes n}
    \right]                                                   \\
    &=
    \int_{\mathcal A}
    \ell(\theta,g^{-1}\omega)
    \operatorname{Tr}\!\left[
        (U_g^{\otimes n}\rho_\theta^{\otimes n}U_g^{\otimes n\dagger})
        N(d\omega)
    \right]                                                   \\
    &=
    \int_{\mathcal A}
    \ell(g\theta,\omega)
    \operatorname{Tr}\!\left[
        \rho_{g\theta}^{\otimes n}N(d\omega)
    \right]                                                   
    =
    R_{g\theta}(N).
\end{aligned}
\]
Averaging this identity over \(\theta\) and using invariance of
\(\mu_\Theta\), we get
\[
    R_{\mu_\Theta}(N_g)
    =
    \int_\Theta R_{g\theta}(N)\,\mu_\Theta(d\theta)
    =
    R_{\mu_\Theta}(N).
\]

Now define the averaged POVM
\[
    \overline N(B)
    :=
    \int_G N_g(B)\,\mu_G(dg),
    \label{eq:averaged_decision_povm_def}
\]
where \(\mu_G\) is the normalized Haar measure on the compact group \(G\).  The
operator integral is well defined in the weak sense; positivity and
normalization follow from positivity and normalization of each \(N_g\).
Linearity of the risk in the POVM gives
\[
\begin{aligned}
    R_{\mu_\Theta}(\overline N)
    &=
    \int_G R_{\mu_\Theta}(N_g)\,\mu_G(dg) 
    =
    R_{\mu_\Theta}(N),
\end{aligned}
\]
where the last equality is the preceding equality.  Hence averaging
does not change the Haar-Bayes risk.

It remains to verify covariance.  For \(h\in G\),
\[
\begin{aligned}
    \overline N(hB)
    &=
    \int_G
    U_g^{\otimes n\dagger}N(ghB)U_g^{\otimes n}
    \,\mu_G(dg).
\end{aligned}
\]
Set \(k=gh\).  By right invariance of Haar measure, \(\mu_G(dg)=\mu_G(dk)\), and
\[
\begin{aligned}
    \overline N(hB)
    &=
    \int_G
    U_{kh^{-1}}^{\otimes n\dagger}N(kB)U_{kh^{-1}}^{\otimes n}
    \,\mu_G(dk)                                                \\
    &=
    U_h^{\otimes n}
    \left(
        \int_G U_k^{\otimes n\dagger}N(kB)U_k^{\otimes n}
        \,\mu_G(dk)
    \right)
    U_h^{\otimes n\dagger}                                      \\
    &=
    U_h^{\otimes n}\overline N(B)U_h^{\otimes n\dagger}.
\end{aligned}
\]
Thus \(\overline N\) is covariant in the sense of
the covariance condition displayed above.

Finally, suppose \(N\) is covariant.  Then the first calculation above, applied
to \(N\) itself, gives
\[
    R_{g\theta}(N)=R_\theta(N)
\]
for all \(g\in G\).  Since \(G\) acts transitively on \(\Theta\), the risk is
constant on \(\Theta\).  Hence
\[
    R_{\max}(N)=R_{\mu_\Theta}(N).
\]
Combining this with the averaging construction proves that the Haar-Bayes
optimization may be restricted to covariant decision POVMs, and that for
covariant procedures the Haar-Bayes and worst-case risks coincide.

\end{proof}

\subsection{Depolarizing form of an expected covariant decision}
\label{subsec:covariant_expected_decision}

This subsection records two consequences of covariance that will be used
separately below.  The first concerns the input-fixed average of a covariant
decision.  The second concerns loss of information when a measurement is
coarse-grained by an equivariant statistic.  The latter formulation includes
subgroup-orbit coarse-graining, but does not require every fiber of the
statistic to be a single group orbit.

Let \(M\) be a measurement with outcome space \(\mathsf X\), and assume that
\(U(d)\) acts measurably on \(\mathsf X\).  We write its covariance as
\begin{equation}
    p_M(d(gx)\mid g\psi)=p_M(dx\mid\psi).
    \label{eq:covariant_observation_model_expected_decision}
\end{equation}
The posterior predictive state is
\begin{equation}
    \rho_M(x)
    :=
    \int \rho_\phi\,\pi_M(d\phi\mid x),
    \label{eq:posterior_state_covariant_decision_section}
\end{equation}
and Haar invariance implies the pointwise covariance relation
\begin{equation}
    \rho_M(gx)=U_g\rho_M(x)U_g^\dagger.
    \label{eq:posterior_state_pointwise_covariance}
\end{equation}
More generally, let \(x\mapsto\sigma(x)\) be a pointwise covariant announced
state:
\begin{equation}
    \sigma(gx)=U_g\sigma(x)U_g^\dagger.
    \label{eq:pointwise_covariant_decision_rule}
\end{equation}
For a fixed true state
\(\rho_\star=\ket{\psi_\star}\bra{\psi_\star}\), define
\begin{equation}
    \overline\rho_{\sigma}^{\mathrm{in}}(\psi_\star)
    :=
    \mathbb E_{X\mid\psi_\star}[\sigma(X)].
    \label{eq:expected_announced_state_fixed_input}
\end{equation}

\begin{lemma}[Depolarizing form of the expected decision]
\label{lem:depolarizing_form_expectation_general}
The input-fixed expected announced state has the form
\begin{equation}
\begin{split}
    \overline\rho_{\sigma}^{\mathrm{in}}(\psi_\star)
    &=
    \tau_\sigma\rho_\star
    +
    \frac{1-\tau_\sigma}{d-1}(I-\rho_\star)                                    \\
    &=
    \alpha_\sigma\rho_\star
    +
    \beta_\sigma\frac{I}{d},
    \qquad
    \alpha_\sigma+\beta_\sigma=1,
\end{split}
\label{eq:expected_decision_two_depolarizing_parameterizations}
\end{equation}
where
\begin{equation}
    \tau_\sigma
    =
    \operatorname{Tr}\!\left[
       \rho_\star\overline\rho_{\sigma}^{\mathrm{in}}(\psi_\star)
    \right]
    =
    \alpha_\sigma+\frac{\beta_\sigma}{d}.
    \label{eq:expected_decision_tau_definition_local}
\end{equation}
These scalars are independent of the true ray.
\end{lemma}

\begin{proof}
Covariance gives
\begin{equation}
    \overline\rho_{\sigma}^{\mathrm{in}}(g\psi_\star)
    =
    U_g\overline\rho_{\sigma}^{\mathrm{in}}(\psi_\star)U_g^\dagger.
    \label{eq:expected_decision_covariance_proof}
\end{equation}
If \(h\) stabilizes the ray of \(\psi_\star\), the left-hand side is
unchanged.  Hence the expected state commutes with the stabilizer, which acts
as the full unitary group on \(\psi_\star^\perp\).  Schur's lemma gives a
scalar on \(\mathbb C\psi_\star\) and another scalar on
\(\psi_\star^\perp\).  The trace-one condition gives the first representation
in \eqref{eq:expected_decision_two_depolarizing_parameterizations}; the second
is its rewriting in the basis \(\{\rho_\star,I/d\}\).  Transitivity makes the
coefficients independent of the true ray.
\end{proof}

We next formulate coarse-graining by a statistic.  Let
\begin{equation}
    T:\mathsf X\longrightarrow\mathsf Y,
    \qquad
    Y=T(X),
    \label{eq:equivariant_statistic_coarse_graining_map}
\end{equation}
be measurable and equivariant:
\begin{equation}
    T(gx)=gT(x).
    \label{eq:equivariant_statistic_condition}
\end{equation}
The push-forward measurement is denoted by \(M_T\).  Its posterior predictive
state is
\begin{equation}
\begin{split}
    \rho_{M_T}(y)
    &:={\mathbb E}[\rho_\Phi\mid Y=y]                                          \\
    &={\mathbb E}[\rho_M(X)\mid T(X)=y].
\end{split}
\label{eq:coarse_posterior_state_conditional_expectation}
\end{equation}
The second equality is the tower property.  If \(T\) is the quotient by a
compact subgroup \(H\), and the conditional law within the fiber is normalized
Haar measure, this conditional expectation reduces to the usual \(H\)-twirl.
The conditional-expectation formulation applies to an arbitrary measurable
equivariant statistic and does not require its fibers to be individual group
orbits.

Define
\begin{equation}
    R(M):=\mathbb E_X S(\rho_M(X)),
    \hspace{5pt}
    R(M_T):=\mathbb E_Y S(\rho_{M_T}(Y)).
    \label{eq:fine_coarse_posterior_entropy_risks}
\end{equation}

\begin{proposition}[Risk increase under equivariant coarse-graining]
\label{prop:risk_increase_under_orbit_coarse_graining}
Every measurable coarse-graining satisfies
\begin{equation}
    R(M)\le R(M_T).
    \label{eq:fine_risk_le_coarse_risk}
\end{equation}
Moreover,
\begin{equation}
    R(M_T)-R(M)
    =
    \mathbb E_X
    D\!\left(
       \rho_M(X)\middle\|\rho_{M_T}(T(X))
    \right).
    \label{eq:coarse_graining_entropy_gap_relative_entropy}
\end{equation}
\end{proposition}

\begin{proof}
By \eqref{eq:coarse_posterior_state_conditional_expectation} and entropy
concavity,
\[
 S(\rho_{M_T}(Y))
 \ge {\mathbb E}[S(\rho_M(X))\mid Y].
\]
Averaging gives \eqref{eq:fine_risk_le_coarse_risk}.  Furthermore,
\[
\begin{aligned}
&{\mathbb E}\!\left[
 D\!\left(\rho_M(X)\middle\|\rho_{M_T}(Y)\right)\middle|Y
\right]\\
&\qquad=
S(\rho_{M_T}(Y))-{\mathbb E}[S(\rho_M(X))\mid Y],
\end{aligned}
\]
which gives the exact gap after averaging.
\end{proof}

Suppose now that \(\mathsf Y=\mathbb P(\mathcal H)\) and that the group acts
transitively.  Write a coarse outcome as a ray \(y\), with projector
\(\rho_y\).  The stabilizer of \(y\) acts irreducibly on \(y^\perp\), while the
covariance of \(M_T\) implies
\[
    U_h\rho_{M_T}(y)U_h^\dagger=\rho_{M_T}(y)
\]
for every stabilizer element \(h\).  Hence
\begin{equation}
    \rho_{M_T}(y)
    =
    r_T\rho_y+
    \frac{1-r_T}{d-1}(I-\rho_y),
    \label{eq:coarse_posterior_depolarizing_form}
\end{equation}
where transitivity makes \(r_T\) independent of \(y\).  When
\(r_T\ge 1/d\), it is the principal eigenvalue.  Setting
\begin{equation}
    h_d(r)
    :=
    -r\log r-(1-r)\log\frac{1-r}{d-1},
    \label{eq:h_d_lambda_definition}
\end{equation}
we obtain
\begin{equation}
    R(M)\le R(M_T)=h_d(r_T).
    \label{eq:fine_risk_upper_bounded_by_coarse_lambda_entropy}
\end{equation}
The parameter has the fidelity representation
\begin{equation}
\begin{split}
    r_T
    &=
    \operatorname{Tr}[\rho_y\rho_{M_T}(y)]                                    \\
    &=
    \int
    \operatorname{Tr}[\rho_\phi\rho_y] \,\pi_{M_T}(d\phi\mid y),
\end{split}
\label{eq:coarse_lambda_output_fidelity}
\end{equation}
which is constant in \(y\).  Thus the reported ray associated with the coarse
outcome determines the depolarizing parameter and therefore the complete
coarse posterior entropy.

For reference, the purity of the coarse posterior is
\begin{equation}
    \operatorname{Tr}[\rho_{M_T}(y)^2]
    =r_T^2+\frac{(1-r_T)^2}{d-1}.
    \label{eq:coarse_posterior_purity_as_p_lambda}
\end{equation}
Averaging and using the posterior-mean identity gives
\begin{equation}
\begin{split}
&\mathbb E_Y\operatorname{Tr}[\rho_{M_T}(Y)^2]                                  \\
&\quad=
\mathbb E_{\Phi,Y}\operatorname{Tr}[\rho_\Phi\rho_{M_T}(Y)]                   \\
&\quad=
\mathbb E_{Y\mid\psi_\star}
\operatorname{Tr}[\rho_\star\rho_{M_T}(Y)] .
\end{split}
\label{eq:coarse_mean_purity_fixed_input_fidelity_identity}
\end{equation}
This purity identity is distinct from the entropy upper bound
\eqref{eq:fine_risk_upper_bounded_by_coarse_lambda_entropy}.

For later use, it is important that coarse-graining moves entropy and purity in
opposite directions.  Define the posterior-overlap functional
\begin{equation}
    \mathcal P(M)
    :=
    \mathbb E_{\Phi,X}
    \operatorname{Tr}[\rho_\Phi\rho_M(X)].
    \label{eq:posterior_overlap_functional_definition}
\end{equation}
By output-first conditioning and the posterior-mean identity,
\begin{equation}
    \mathcal P(M)
    =
    \mathbb E_X\operatorname{Tr}[\rho_M(X)^2].
    \label{eq:posterior_overlap_equals_mean_purity}
\end{equation}
Thus this is simultaneously the joint overlap between the latent pure state
and the posterior predictive state and the mean purity of that posterior
state.

\begin{proposition}[Purity loss under coarse-graining]
\label{prop:purity_loss_under_coarse_graining}
For every measurable statistic \(T\),
\begin{equation}
    \mathcal P(M_T)\le \mathcal P(M).
    \label{eq:coarse_graining_decreases_posterior_purity}
\end{equation}
More precisely,
\begin{equation}
\begin{split}
&\mathcal P(M)-\mathcal P(M_T)\\
&\quad=
\mathbb E_X
\operatorname{Tr}\!\left[
 \left(
   \rho_M(X)-\rho_{M_T}(T(X))
 \right)^2
\right].
\end{split}
\label{eq:coarse_graining_purity_gap_hilbert_schmidt}
\end{equation}
\end{proposition}

\begin{proof}
Write \(Y=T(X)\).  By
\eqref{eq:coarse_posterior_state_conditional_expectation},
\(
 \rho_{M_T}(Y)=\mathbb E[\rho_M(X)\mid Y]
\).
Expanding the square conditionally on \(Y\), and using the defining property of
conditional expectation, gives
\[
\begin{aligned}
&\mathbb E\!\left[
 \operatorname{Tr}
 \left(
   \rho_M(X)-\rho_{M_T}(Y)
 \right)^2
 \middle|Y
\right]\\
&\qquad=
\mathbb E[\operatorname{Tr}(\rho_M(X)^2)\mid Y]
-\operatorname{Tr}(\rho_{M_T}(Y)^2).
\end{aligned}
\]
Averaging and applying
\eqref{eq:posterior_overlap_equals_mean_purity} proves both statements.
\end{proof}

Consequently, replacing a measurement without a useful outcome stabilizer by
an equivariant ray-valued coarse measurement can expose a one-parameter
spectral description, but it cannot prove an upper bound on the original fine
measurement's posterior-overlap functional: the inequality points in the
opposite direction.  For a transitive ray-valued coarse measurement,
\eqref{eq:coarse_posterior_depolarizing_form} gives
\begin{equation}
    \mathcal P(M_T)
    =
    p_d(r_T),
    \qquad
    p_d(r):=r^2+\frac{(1-r)^2}{d-1}.
    \label{eq:coarse_posterior_overlap_as_fidelity_function}
\end{equation}
Moreover,
\begin{equation}
    p_d'(r)=\frac{2(dr-1)}{d-1}.
    \label{eq:coarse_purity_fidelity_monotonicity_derivative}
\end{equation}
Hence, on \(r\ge1/d\), maximizing the coarse posterior purity is equivalent to
maximizing the reported-ray fidelity \(r_T\).  This equivalence is useful for
interpreting projective-space coarse measurements, but the unrestricted global
optimization below will be proved directly from the arbitrary-effect posterior
identity rather than by coarse-graining.

\section{Global Optimization over All Measurements}  

\label{sec:global_optimization}

This section treats the collective prediction problem in which the first
\(n\) copies are measured jointly by an arbitrary POVM, while one additional
copy in the same unknown pure state is left unmeasured as the prediction target.
The proof is direct and outcome-wise.  It does not restrict the measurement
outcome space, invoke a seed representation, or reduce general density-matrix
orbits to projective-space outcomes.
The Hunt--Stein reduction in the previous section is retained for
Bayes--minimax comparisons and for the covariance arguments used later.
The Bayes optimality proof in the present global section is instead the
direct effect-wise majorization argument below.

\subsection{Symmetric-subspace notation and arbitrary POVM effects}
\label{subsec:global_symmetric_power_notation}
Let \(\mathcal H\simeq\mathbb C^d\), let
\[
    \mathcal K_n:=\operatorname{Sym}^n(\mathcal H),
    \quad
    D_n:=\dim\mathcal K_n=\binom{n+d-1}{d-1},
\]
and write \(\Pi_{\mathrm{sym}}^{(n)}\) for the orthogonal projection onto
\(\mathcal K_n\).  The observed state corresponding to a pure ray \(\phi\) is
\(\rho_\phi^{\otimes n}\), which is supported on \(\mathcal K_n\).
Therefore every POVM may, without changing its outcome probabilities, be
compressed to \(\mathcal K_n\).
Operators on \(\mathcal K_n\) are henceforth identified with their zero
extensions to \(\mathcal H^{\otimes n}\).  In particular, permutations of
the observed tensor factors act as the identity on the support of every
such compressed operator.

Let \(M_n(B):=\Pi_{\mathrm{sym}}^{(n)}M(B)\Pi_{\mathrm{sym}}^{(n)}\) be the
compressed POVM and define the finite scalar measure
\begin{equation}
    \lambda_M(B):=\operatorname{Tr}M_n(B).
    \label{eq:global_povm_scalar_dominating_measure}
\end{equation}
Because \(\mathcal K_n\) is finite dimensional, every matrix element of
\(M_n\) is absolutely continuous with respect to \(\lambda_M\).  The scalar
Radon--Nikodym derivatives assemble into a weakly measurable operator density
\(E_x\) satisfying
\begin{align}
    M_n(dx)&=E_x\,\lambda_M(dx),\nonumber\\
    &\quad\qquad E_x\ge0
    \quad \lambda_M\text{-almost everywhere}.
    \label{eq:global_povm_radon_nikodym_density}
\end{align}
Positivity follows by applying the scalar Radon--Nikodym theorem to
\(\langle v,M_n(\cdot)v\rangle\) for every \(v\in\mathcal K_n\), and
polarization identifies the common operator density.  Moreover,
\begin{equation}
    \int E_x\,\lambda_M(dx)=I_{\mathcal K_n}
    \label{eq:global_povm_density_weak_normalization}
\end{equation}
in the weak operator sense.  With the choice
\eqref{eq:global_povm_scalar_dominating_measure}, one may take
\(\operatorname{Tr}E_x=1\) almost everywhere after discarding a
\(\lambda_M\)-null set.  Thus it is enough to analyze the posterior state
generated by one arbitrary nonzero positive effect \(E\) on
\(\mathcal K_n\).

The Haar tensor identity used below is
\begin{equation}
    \int_{\mathbb P(\mathcal H)}
    \rho_\phi^{\otimes t}\,\nu(d\phi)
    =
    \frac{\Pi_{\mathrm{sym}}^{(t)}}{D_t},
    \qquad
    D_t=\binom{t+d-1}{d-1}.
    \label{eq:global_haar_tensor_identity}
\end{equation}

\subsection{Posterior predictive state associated with an arbitrary effect}
\label{subsec:global_arbitrary_effect_posterior}

\begin{lemma}[One-particle reduction of the extended symmetric projector]
\label{lem:global_symmetric_projector_partial_trace}
For every operator \(E\) supported on \(\mathcal K_n\),
\begin{equation}
\begin{split}
&\operatorname{Tr}_{1,\ldots,n}
\left[
    (E\otimes I)\Pi_{\mathrm{sym}}^{(n+1)}
\right]\\
&\qquad=
\frac{
    (\operatorname{Tr}E)I
    +n\operatorname{Tr}_{2,\ldots,n}E
}{n+1}.
\end{split}
\label{eq:global_symmetric_projector_partial_trace}
\end{equation}
Here the one-particle operator
\(\operatorname{Tr}_{2,\ldots,n}E\) is identified with an operator on the
untraced \((n+1)\)-st copy.
\end{lemma}

\begin{proof}
The symmetric projector has the coset decomposition
\begin{equation}
    \Pi_{\mathrm{sym}}^{(n+1)}
    =
    \frac{1}{n+1}
    \left(I+\sum_{j=1}^{n}S_{j,n+1}\right)
    \left(\Pi_{\mathrm{sym}}^{(n)}\otimes I\right),
    \label{eq:global_symmetric_projector_coset_decomposition}
\end{equation}
where \(S_{j,n+1}\) exchanges the \(j\)-th observed factor with the future
factor.  Since \(E\) is supported on \(\mathcal K_n\),
\[
    E\Pi_{\mathrm{sym}}^{(n)}
    =\Pi_{\mathrm{sym}}^{(n)}E=E.
\]
The identity term in \eqref{eq:global_symmetric_projector_coset_decomposition}
therefore contributes \((\operatorname{Tr}E)I\). We use the convention that
\(S_{j,n+1}\) exchanges the \(j\)-th and \((n+1)\)-st tensor factors. Letting 
\(\mathbf{i}_L = i_1 \cdots i_{j-1}\) and \(\mathbf{i}_R = i_{j+1} \cdots i_n\) 
denote the tensor indices to the left and right of the \(j\)-th position, 
for an orthonormal basis \(\{\ket a\}\), its partial-SWAP action is fixed by
\begin{align}
&\bra{a}\operatorname{Tr}_{1,\ldots,n}[S_{j,n+1}(E\otimes I)]\ket{b} \nonumber\\
&\qquad\qquad\qquad\qquad= \sum_{\mathbf{i}_L, \mathbf{i}_R} 
\langle \mathbf{i}_L, a, \mathbf{i}_R \rvert E \lvert \mathbf{i}_L, b, \mathbf{i}_R \rangle.
\label{eq:global_partial_swap_matrix_element}
\end{align}
Thus the \(j\)-th swap term is precisely the one-particle marginal of
\(E\) in the \(j\)-th tensor position, with no transpose or complex
conjugation. Moreover, under the zero-extension convention above, every permutation
\(P_\pi\) of the \(n\) observed factors satisfies
\[
    P_\pi E=E=EP_\pi,
\]
because \(P_\pi\) acts as the identity on \(\mathcal K_n\). Hence all
one-particle marginals of \(E\) coincide. Each of the \(n\) swap terms is thus
\(\operatorname{Tr}_{2,\ldots,n}E\). Dividing the sum by \(n+1\) proves
\eqref{eq:global_symmetric_projector_partial_trace}.
\end{proof}

\begin{lemma}[Posterior predictive state associated with an arbitrary effect]
\label{lem:global_arbitrary_effect_posterior}
Let \(E\ne0\) be a positive operator on \(\mathcal K_n\), and define its
normalized one-particle marginal by
\begin{equation}
    \rho_E
    :=
    \frac{\operatorname{Tr}_{2,\ldots,n}E}{\operatorname{Tr}E}.
    \label{eq:global_effect_one_particle_marginal}
\end{equation}
Then \(\rho_E\) is a density operator on \(\mathcal H\), and the posterior
predictive state conditioned on the effect \(E\) is
\begin{equation}
    \mu_E
    =
    \frac{I+n\rho_E}{n+d}.
    \label{eq:global_arbitrary_effect_posterior}
\end{equation}
In particular,
\begin{equation}
    \mu_E\ge\frac{I}{n+d}.
    \label{eq:global_arbitrary_effect_spectral_floor}
\end{equation}
\end{lemma}

\begin{proof}
The prior-predictive probability density associated with \(E\) is
\begin{equation}
    p_E
    :=
    \int
    \operatorname{Tr}\!\left[E\rho_\phi^{\otimes n}\right]\nu(d\phi)
    =
    \frac{\operatorname{Tr}E}{D_n},
    \label{eq:global_effect_prior_predictive_probability}
\end{equation}
where we used \eqref{eq:global_haar_tensor_identity}.  The corresponding
unnormalized posterior predictive state is
\begin{equation}
\begin{split}
    A_E
    &:=
    \int
    \rho_\phi\,
    \operatorname{Tr}\!\left[E\rho_\phi^{\otimes n}\right]\nu(d\phi)\\
    &=
    \frac{1}{D_{n+1}}
    \operatorname{Tr}_{1,\ldots,n}
    \left[
        (E\otimes I)\Pi_{\mathrm{sym}}^{(n+1)}
    \right].
\end{split}
\label{eq:global_effect_unnormalized_posterior}
\end{equation}
Lemma~\ref{lem:global_symmetric_projector_partial_trace} and
\[
    \frac{D_n}{D_{n+1}}=\frac{n+1}{n+d}
\]
now give
\[
\begin{aligned}
    \mu_E
    =\frac{A_E}{p_E}
    &=
    \frac{D_n}{D_{n+1}}
    \frac{(\operatorname{Tr}E)I
    +n\operatorname{Tr}_{2,\ldots,n}E}
    {(n+1)\operatorname{Tr}E}\\
    &=
    \frac{I+n\rho_E}{n+d}.
\end{aligned}
\]
Positivity of \(E\) implies positivity of its partial trace, and
\(\operatorname{Tr}\rho_E=1\); hence \(\rho_E\) is a density operator.  The
spectral lower bound follows immediately from \(\rho_E\ge0\).
\end{proof}

The identity \eqref{eq:global_arbitrary_effect_posterior} is not merely an
evaluation of the eventual optimal covariant POVM.  It applies to every nonzero
positive effect in every collective POVM after compression to
\(\mathcal K_n\), and shows that the effect-conditioned predictive problem
for the future copy depends on the full \(n\)-copy effect only through its
normalized one-particle marginal.  Consequently, the measurement optimization
below can be reduced outcome by outcome to a spectral optimization over
one-copy density operators.

\begin{remark}[Consistency checks]
For \(E=I_{\mathcal K_n}\), one has \(\rho_E=I/d\) and hence
\(\mu_E=I/d\).  For the coherent rank-one effect
\(E=(\ket\psi\bra\psi)^{\otimes n}\), one has
\(\rho_E=\ket\psi\bra\psi\) and
\[
    \mu_E=\frac{I+n\ket\psi\bra\psi}{n+d}.
\]
For \(n=1\), this reduces to
\((I+\ket\psi\bra\psi)/(d+1)\).
\end{remark}

\begin{lemma}[Characterization of equality effects]
\label{lem:global_pure_marginal_coherent_effect}
Let \(E\ne0\) be positive and supported on \(\mathcal K_n\).  Its normalized
one-particle marginal \(\rho_E\) is pure if and only if there are a ray
\(\psi\) and a scalar \(c>0\) such that
\begin{equation}
    E=c(\ket\psi\bra\psi)^{\otimes n}.
    \label{eq:global_equality_effect_coherent_form}
\end{equation}
\end{lemma}
\begin{proof}
The reverse implication is immediate.  Conversely, suppose
\(\rho_E=\ket\psi\bra\psi\), and put
\(Q:=I-\ket\psi\bra\psi\).  On \(\mathcal H^{\otimes n}\), define the positive
operator
\[
    A:=\sum_{j=1}^n Q_j,
\]
where \(Q_j\) acts as \(Q\) on the \(j\)-th factor.  Since all one-particle
marginals of an operator supported on \(\mathcal K_n\) coincide,
\[
    \operatorname{Tr}(EA)
    =n\operatorname{Tr}(E)\operatorname{Tr}(\rho_E Q)=0.
\]
Both \(E\) and \(A\) are positive.  Hence
\(E^{1/2}AE^{1/2}=0\), so the support of \(E\) is contained in
\(\ker A\).  Within the symmetric subspace, \(\ker A\) is the one-dimensional
space spanned by \(\ket\psi^{\otimes n}\): a symmetric tensor has zero
occupation in \(\psi^\perp\) exactly when every tensor factor lies along
\(\psi\).  Therefore \(E\) is a positive scalar multiple of the projector onto
that vector, which is \eqref{eq:global_equality_effect_coherent_form}.
\end{proof}

\subsection{Pointwise majorization and entropy lower bound}
\label{subsec:global_pointwise_majorization}

\begin{lemma}[Pointwise entropy lower bound for arbitrary effects]
\label{lem:global_pointwise_entropy_lower_bound}
Let \(E\ne0\) be any positive effect on \(\mathcal K_n\).  For any pure state
\(\ket\psi\), define
\begin{equation}
    \mu_{\mathrm{hw}}(\psi)
    :=
    \frac{I+n\ket\psi\bra\psi}{n+d}.
    \label{eq:global_highest_weight_posterior_reference}
\end{equation}
Then the spectrum of \(\mu_{\mathrm{hw}}(\psi)\) majorizes the spectrum of
\(\mu_E\), and therefore
\begin{equation}
    S(\mu_E)
    \ge
    S\!\left(\mu_{\mathrm{hw}}(\psi)\right).
    \label{eq:global_pointwise_entropy_bound}
\end{equation}
The right-hand side is independent of \(\psi\).
\end{lemma}

\begin{proof}
Let \(r=(r_1,\ldots,r_d)\) be the decreasingly ordered eigenvalue vector of
\(\rho_E\).  By Lemma~\ref{lem:global_arbitrary_effect_posterior},
\begin{equation}
    \lambda(\mu_E)
    =
    \frac{1}{n+d}(1,\ldots,1)
    +
    \frac{n}{n+d}r.
    \label{eq:global_arbitrary_effect_spectrum_affine}
\end{equation}
The eigenvalue vector of a pure state, \(e_1=(1,0,\ldots,0)\), majorizes every
probability vector, so \(e_1\succ r\).  Majorization is preserved under
multiplication by a nonnegative scalar and addition of a common vector.  Hence
\begin{equation}
\begin{split}
    \lambda\!\left(\mu_{\mathrm{hw}}(\psi)\right)
    &=
    \frac{1}{n+d}(1,\ldots,1)
    +
    \frac{n}{n+d}e_1\\
    &\succ
    \frac{1}{n+d}(1,\ldots,1)
    +
    \frac{n}{n+d}r
    =\lambda(\mu_E).
\end{split}
\label{eq:global_pointwise_majorization}
\end{equation}
The von Neumann entropy is Schur concave, which proves
\eqref{eq:global_pointwise_entropy_bound}.
\end{proof}

\subsection{Global optimality for posterior overlap and mean purity}
\label{subsec:global_posterior_overlap_optimality}

The arbitrary-effect identity also solves a second measurement optimization.
For a collective POVM \(M\) on the \(n\) observed copies, let \(X\) denote its
outcome and let \(\mu_M(X)\) be the posterior predictive state.  Consider
\begin{equation}
\begin{split}
    \mathcal P_n(M)
    &:={\mathbb E}_{\Phi,X}
      \operatorname{Tr}[\rho_\Phi\mu_M(X)]\\
    &={\mathbb E}_X\operatorname{Tr}[\mu_M(X)^2].
\end{split}
\label{eq:global_posterior_overlap_criterion}
\end{equation}
The second equality is
\eqref{eq:posterior_overlap_equals_mean_purity}.  Unlike the relative-entropy
risk, this quadratic criterion is to be maximized.

\begin{proposition}[Global optimality for posterior overlap]
\label{prop:global_posterior_overlap_optimality}
For every collective POVM \(M\) on the \(n\) observed copies,
\begin{equation}
    \mathcal P_n(M)
    \le
    \frac{(n+1)^2+d-1}{(n+d)^2}.
    \label{eq:global_posterior_overlap_upper_bound}
\end{equation}
The highest-weight covariant POVM
\eqref{eq:highest_weight_collective_povm} attains equality outcome by outcome.
Consequently,
\begin{equation}
    \sup_M \mathcal P_n(M)
    =
    \frac{(n+1)^2+d-1}{(n+d)^2}.
    \label{eq:global_posterior_overlap_optimal_value}
\end{equation}
\end{proposition}

\begin{proof}
Let \(E\ne0\) be any positive outcome effect after compression to
\(\mathcal K_n\).  By
Lemma~\ref{lem:global_arbitrary_effect_posterior},
\(
 \mu_E=(I+n\rho_E)/(n+d)
\), where \(\rho_E\) is a density operator.  Therefore
\begin{equation}
\begin{split}
    \operatorname{Tr}(\mu_E^2)
    &=
    \frac{
      d+2n+n^2\operatorname{Tr}(\rho_E^2)
    }{(n+d)^2}\\
    &\le
    \frac{d+2n+n^2}{(n+d)^2}
    =
    \frac{(n+1)^2+d-1}{(n+d)^2},
\end{split}
\label{eq:global_effect_purity_upper_bound}
\end{equation}
where we used \(\operatorname{Tr}(\rho_E^2)\le1\).  The bound is independent of
the outcome, so averaging over the prior-predictive outcome distribution proves
\eqref{eq:global_posterior_overlap_upper_bound}.

For a coherent effect
\(
 E=(\ket\psi\bra\psi)^{\otimes n}
\), its normalized one-particle marginal is
\(\rho_E=\ket\psi\bra\psi\), and hence
\(\operatorname{Tr}(\rho_E^2)=1\).  Every effect density of the highest-weight
POVM is coherent, so the pointwise bound is attained at every outcome.  This
proves \eqref{eq:global_posterior_overlap_optimal_value}.
\end{proof}

\begin{remark}[Relation to estimation fidelity and coarse-graining]
\label{rem:global_overlap_fidelity_and_coarse_graining}
For the highest-weight measurement, the reported outcome is a ray
\(\widehat\psi\), and
\(
 r_{n,d}^{\mathrm{glob}}=(n+1)/(n+d)
\)
is the mean fidelity obtained by reporting that ray.  The optimized criterion
in Proposition~\ref{prop:global_posterior_overlap_optimality} is instead the
purity of the mixed posterior predictive state:
\begin{equation}
\begin{split}
    \mathcal P_n^{\mathrm{glob}}
    &=
    \left(r_{n,d}^{\mathrm{glob}}\right)^2
    +\frac{\left(1-r_{n,d}^{\mathrm{glob}}\right)^2}{d-1}\\
    &=
    \frac{(n+1)^2+d-1}{(n+d)^2}.
\end{split}
\label{eq:global_overlap_value_as_fidelity_function}
\end{equation}
Thus reported-ray fidelity and posterior overlap are distinct quantities, but
for a stabilizer-reduced ray-valued measurement the latter is the increasing
function \(p_d(r)\) of the former on \(r\ge1/d\).  For a general measurement
without such an outcome stabilizer, one may coarse-grain to obtain this
interpretation, but
Proposition~\ref{prop:purity_loss_under_coarse_graining} shows that the
coarse-grained value is no larger than the original one.  The unrestricted
global upper bound therefore comes from the direct arbitrary-effect calculation
\eqref{eq:global_effect_purity_upper_bound}, not from the coarse-graining step.
\end{remark}

\subsection{Global Bayes and minimax optimality}
\label{subsec:global_bayes_minimax_optimality}

The highest-weight covariant POVM is
\begin{equation}
    \Pi_n(d\hat\psi)
    =
    D_n
    (\ket{\hat\psi}\bra{\hat\psi})^{\otimes n}
    \nu(d\hat\psi).
    \label{eq:highest_weight_collective_povm}
\end{equation}
Its normalization follows directly from
\eqref{eq:global_haar_tensor_identity} with \(t=n\).  Coherent or
highest-weight measurements are standard in fidelity-based pure-state
estimation
\cite{MassarPopescu1995OptimalEstimation,DerkaBuzekEkert1998OptimalEstimation,Hayashi1998PureStateEstimation,Group2}.
The theorem below addresses a different decision problem: the reported object
is the mixed posterior predictive state, the loss is quantum relative entropy,
and optimality over all POVMs follows from the arbitrary-effect identity and a
pointwise entropy bound rather than from a fidelity merit theorem.

\begin{theorem}[Global highest-weight optimality]
\label{thm:highest_weight_global_optimality}
In the Haar-prior \(n\)-observed-copy pure-state prediction problem under
quantum relative-entropy loss, the minimum Bayes risk over all collective POVMs
and all density-matrix-valued decision rules is attained by the highest-weight
covariant POVM \eqref{eq:highest_weight_collective_povm}, followed by the
posterior predictive-state rule
\begin{equation}
    \mu_n^{\mathrm{glob}}(\hat\psi)
    =
    \frac{n\ket{\hat\psi}\bra{\hat\psi}+I}{n+d}.
    \label{eq:global_mu_explicit_final}
\end{equation}
\end{theorem}

\begin{proof}
Fix an arbitrary collective POVM \(M\).  Using
\eqref{eq:global_povm_radon_nikodym_density}, let \(E_x\) denote its compressed
positive density for \(\lambda_M\)-almost every outcome.  By
Theorem~\ref{thm:fixed_measurement_bayes_act}, the optimal announced state at
that outcome is the posterior predictive state \(\mu_{E_x}\), and the optimized
conditional loss is \(S(\mu_{E_x})\).

Fix an arbitrary reference ray \(\psi_0\).  Then
Lemma~\ref{lem:global_pointwise_entropy_lower_bound} gives, for every outcome
having nonzero prior-predictive probability,
\[
    S(\mu_{E_x})
    \ge
    S\!\left(\mu_{\mathrm{hw}}(\psi_0)\right).
\]
The lower bound is independent of both the outcome and the chosen reference ray.  Averaging it with respect to
the prior-predictive outcome distribution therefore yields
\begin{equation}
    R(M)
    \ge
    S\!\left(\mu_{\mathrm{hw}}(\psi_0)\right).
    \label{eq:global_all_povm_risk_lower_bound}
\end{equation}
This argument places no restriction on the original outcome space or on the
form of the POVM.

For the POVM \(\Pi_n\), every outcome density is proportional to the coherent
effect \((\ket{\hat\psi}\bra{\hat\psi})^{\otimes n}\).  Its normalized
one-particle marginal is \(\ket{\hat\psi}\bra{\hat\psi}\), so
Lemma~\ref{lem:global_arbitrary_effect_posterior} gives exactly
\eqref{eq:global_mu_explicit_final}.  Thus \(\Pi_n\) attains the pointwise lower
bound in \eqref{eq:global_all_povm_risk_lower_bound} at every outcome.  It
therefore attains the minimum over all POVMs and all subsequent decision rules.
\end{proof}

\begin{corollary}[Global minimax optimality]
\label{cor:global_minimax_optimality}
Let \(R_\phi(M,\sigma)\) denote the frequentist relative-entropy risk at the
true pure state \(\phi\).  Then the highest-weight procedure is minimax over all
collective POVMs and all density-matrix-valued decision rules, and
\begin{equation}
    \inf_{M,\sigma}\sup_\phi R_\phi(M,\sigma)
    =
    R_n^{\mathrm{glob}},
    \label{eq:global_minimax_value}
\end{equation}
where \(R_n^{\mathrm{glob}}\) is given in
\eqref{eq:global_optimized_entropy_risk} below.
\end{corollary}

\begin{proof}
For every procedure,
\[
    \sup_\phi R_\phi(M,\sigma)
    \ge
    \int R_\phi(M,\sigma)\,\nu(d\phi)
    \ge
    R_n^{\mathrm{glob}},
\]
where the second inequality is Theorem~\ref{thm:highest_weight_global_optimality}.
The highest-weight POVM and the posterior rule
\eqref{eq:global_mu_explicit_final} are jointly unitary covariant.  Their
frequentist risk is therefore constant on the transitive pure-state orbit, so
its worst-case value equals its Haar-Bayes value
\(R_n^{\mathrm{glob}}\).  This proves the equality.
\end{proof}

\begin{remark}[Outcome-wise optimization and the equality question] The preceding proof avoids a classification of general POVMs by reducing the optimization at each outcome to a spectral problem for a one-particle density operator. Indeed, the posterior predictive state associated with an arbitrary nonzero positive effect \(E\) depends on \(E\) only through its normalized one-particle marginal: \[ \mu_E=\frac{I+n\rho_E}{n+d}. \] Thus the effect-wise entropy minimization is reduced to minimizing \[ S\!\left(\frac{I+n\rho}{n+d}\right) \] over one-particle density operators \(\rho\). The spectrum of a pure state majorizes that of every density operator, and the von Neumann entropy is Schur concave. Consequently, a coherent effect gives a common lower bound for the optimized conditional loss at every outcome. Since the highest-weight covariant POVM consists entirely of coherent effects, it attains this lower bound outcome by outcome, and hence also after averaging over the prior-predictive outcome distribution. Lemma~\ref{lem:global_pure_marginal_coherent_effect} characterizes the equality effects.  Since the pointwise entropy gap is nonnegative, equality of the averaged risk implies equality almost everywhere with respect to the prior-predictive outcome measure.  Strict Schur concavity then forces \(\rho_E\) to be pure almost everywhere, and every positive symmetric-subspace effect with this property is a positive multiple of a coherent rank-one effect.  Thus every outcome of a globally optimal POVM having nonzero prior-predictive probability is coherent up to a prior-predictive null set. A complete classification of globally optimal POVMs would additionally require classifying all measurable coherent-effect resolutions of the symmetric-subspace identity. The present paper does not claim uniqueness of the continuous highest-weight realization. 
\end{remark}

\subsection{Depolarizing posterior for the globally optimal measurement}
\label{subsec:global_explicit_alpha_beta_tau}

The highest-weight outcome is itself a ray \(\widehat\psi\).  Its stabilizer
acts irreducibly on \(\widehat\psi^\perp\), so the general covariant-outcome
argument of Subsection~\ref{subsec:covariant_expected_decision} applies
directly.  From \eqref{eq:global_mu_explicit_final},
\begin{equation}
\begin{split}
    \rho_{\mathrm{glob}}(\widehat\psi)
    &:={\mu}_n^{\mathrm{glob}}(\widehat\psi)                                   \\
    &=
    r_{n,d}^{\mathrm{glob}}\rho_{\widehat\psi}
    +
    \frac{1-r_{n,d}^{\mathrm{glob}}}{d-1}
    (I-\rho_{\widehat\psi}),
\end{split}
\label{eq:global_posterior_depolarizing_form}
\end{equation}
where
\begin{equation}
    r_{n,d}^{\mathrm{glob}}
    =
    \frac{n+1}{n+d}.
    \label{eq:r_glob_reported_def_value}
\end{equation}
Thus the spectrum is independent of the outcome and equals
\begin{equation}
    \frac{n+1}{n+d},
    \qquad
    \frac{1}{n+d}
    \quad(d-1\text{ times}).
    \label{eq:global_mu_eigenvalues}
\end{equation}
The globally optimized relative-entropy risk is therefore
\begin{equation}
\begin{split}
    R_n^{\mathrm{glob}}
    &=h_d(r_{n,d}^{\mathrm{glob}})                                               \\
    &=-\frac{n+1}{n+d}\log\frac{n+1}{n+d}
      -(d-1)\frac{1}{n+d}\log\frac{1}{n+d}.
\end{split}
\label{eq:global_optimized_entropy_risk}
\end{equation}
The same scalar is also the posterior expected fidelity of reporting the
highest-weight outcome ray:
\begin{equation}
\begin{split}
    r_{n,d}^{\mathrm{glob}}
    &=
    \operatorname{Tr}\!\left[
       \rho_{\widehat\psi}\rho_{\mathrm{glob}}(\widehat\psi)
    \right]                                                                     \\
    &=
    \int
    \operatorname{Tr}[\rho_\phi\rho_{\widehat\psi}]\,
    \pi_{\mathrm{glob}}(d\phi\mid\widehat\psi).
\end{split}
\label{eq:global_parameter_as_posterior_fidelity}
\end{equation}
This reported-ray fidelity should not be confused with the purity of the mixed
posterior state, which is
\begin{equation}
    \operatorname{Tr}[\rho_{\mathrm{glob}}(\widehat\psi)^2]
    =
    \left(r_{n,d}^{\mathrm{glob}}\right)^2
    +
    \frac{\left(1-r_{n,d}^{\mathrm{glob}}\right)^2}{d-1}.
    \label{eq:global_mixed_posterior_purity}
\end{equation}
Equations \eqref{eq:global_posterior_depolarizing_form}--
\eqref{eq:global_optimized_entropy_risk} are the application of the general
depolarizing-posterior framework to the globally optimal measurement.  In
addition, \eqref{eq:global_mixed_posterior_purity} equals the globally maximal
posterior-overlap value in
Proposition~\ref{prop:global_posterior_overlap_optimality}.

\section{Discussion}
\label{sec:discussion}
For a fixed measurement, the posterior predictive state is the Bayes act under
quantum relative-entropy loss, and its optimized conditional loss is the
entropy of that posterior state.  The global problem is solved here by an
arbitrary-effect argument.  An effect-conditioned posterior has the form
\((I+n\rho_E)/(n+d)\), and pointwise majorization shows that coherent effects
minimize its entropy.  The highest-weight covariant POVM attains the resulting
bound outcome by outcome.  The same effect-wise spectral extremality maximizes
\(\mathbb E_{\Phi,Y}\operatorname{Tr}[\rho_\Phi\rho_M(Y)]\), which equals the mean
purity of the posterior predictive state.

The proof separates the unrestricted measurement optimization from covariance.
The Bayes optimum is identified directly for every nonzero positive effect,
without first restricting the outcome space or classifying covariant seeds.
Covariance is then used to show that the explicit highest-weight procedure has
constant frequentist risk and is therefore minimax.  The general compact-group
averaging and coarse-graining results remain useful independently of the direct
global proof: coarse-graining increases posterior entropy and decreases
posterior purity, with exact gap identities in relative entropy and squared
Hilbert--Schmidt distance, respectively.

Equality in the pointwise entropy bound requires the normalized one-particle
marginal to be pure.  Lemma~\ref{lem:global_pure_marginal_coherent_effect}
shows that a positive effect on the symmetric subspace has such a marginal if
and only if it is proportional to a coherent rank-one effect.  Hence the
remaining structural question is not the characterization of individual
equality effects, but the classification of all measurable coherent-effect
resolutions of the symmetric-subspace identity.  We do not claim uniqueness of
the continuous highest-weight realization.

The reported-ray fidelity associated with the highest-weight outcome and the
posterior overlap of the mixed predictive state are distinct quantities.  The
former is \((n+1)/(n+d)\), whereas the latter is the purity of the
identity-regularized posterior and equals
\(((n+1)^2+d-1)/(n+d)^2\).  Their relation is transparent for the ray-valued
highest-weight measurement, but the unrestricted posterior-overlap optimum is
proved directly from the arbitrary-effect calculation rather than inferred
from a coarse-graining.

Extensions to nonuniform priors, mixed-state families, multiple unmeasured
future copies, and other predictive losses would require replacements for the
Haar symmetric-tensor identity and the resulting affine one-particle
reduction.  The companion fixed product-observation problem instead keeps the
measurement fixed and uses the present exact global endpoints as benchmarks.

\section*{Acknowledgements} 
M.H. was supported in part by the General R\&D Projects of 1+1+1
CUHK--CUHK(SZ)--GDST Joint Collaboration Fund (Grant No.~GRDP2025-022) and the
Guangdong Provincial Quantum Science Strategic Initiative (Grant
No.~GDZX2505003).
A.D. was supported by the TCS Research Scholar Program (Cycle 19) from Tata Consultancy Services. N.A.W. acknowledges the funding support from the National Quantum Mission, an initiative of the Department of Science and Technology, Govt. of India and the support provided by the Foundation for QC Innovation (FQCI), DST-NQM T-Hub at IISc Bengaluru, in facilitating this project. N.A.W also acknowledges the support provided by the Grant  ANRF/ARG/2025/012066/MS from the Department of Science \& Technology, Govt of India. 
During the preparation of this manuscript, the authors used Microsoft Copilot with the GPT-5.5 and 5.6 Thinking model to assist with language editing, organization and presentation of the manuscript, and the exploration, development, and checking of certain mathematical derivations and arguments. 
All AI-assisted material was critically reviewed, verified, and revised by the authors, who take full responsibility for the accuracy and integrity of the manuscript.


\begin{thebibliography}{99}
\bibitem{SchackBrunCaves2001QuantumBayes}
Schack, Ruediger and Brun, Todd A. and Caves, Carlton M..
\newblock Quantum Bayes rule.
\newblock \emph{Physical Review A}, 64:014305, 2001.  
\newblock \href{https://doi.org/10.1103/PhysRevA.64.014305}{\texttt{doi:10.1103/PhysRevA.64.014305}}.  
\bibitem{TanakaKomaki2005BayesianPredictive}
Tanaka, Fuyuhiko and Komaki, Fumiyasu.
\newblock Bayesian predictive density operators for exchangeable quantum-statistical models.
\newblock \emph{Physical Review A}, 71:052323, 2005.  
\newblock \href{https://doi.org/10.1103/PhysRevA.71.052323}{\texttt{doi:10.1103/PhysRevA.71.052323}}.  
\bibitem{Aitchison1975PredictionFit}
Aitchison, J..
\newblock Goodness of prediction fit.
\newblock \emph{Biometrika}, 62(3):547--554, 1975.  
\newblock \href{https://doi.org/10.1093/biomet/62.3.547}{\texttt{doi:10.1093/biomet/62.3.547}}.  
\bibitem{AitchisonDunsmore1975Prediction}
Aitchison, J. and Dunsmore, I. R..
\newblock Statistical Prediction Analysis.
\newblock \emph{Cambridge University Press}, 1975.

\bibitem{Geisser1993Predictive}
Geisser, Seymour.
\newblock Predictive Inference: An Introduction.
\newblock \emph{Chapman and Hall}, 1993.  
\newblock \href{https://doi.org/10.1007/978-1-4899-4467-2}{\texttt{doi:10.1007/978-1-4899-4467-2}}.  
\bibitem{Komaki2006ShrinkagePrediction}
Komaki, Fumiyasu.
\newblock Shrinkage priors for Bayesian prediction.
\newblock \emph{Annals of Statistics}, 34(2):808--819, 2006.  
\newblock \href{https://doi.org/10.1214/009053606000000010}{\texttt{doi:10.1214/009053606000000010}}.  
\bibitem{GeorgeLiangXu2006PredictiveKL}
George, Edward I. and Liang, Feng and Xu, Xinyi.
\newblock Improved minimax predictive densities under Kullback--Leibler loss.
\newblock \emph{Annals of Statistics}, 34(1):78--91, 2006.  
\newblock \href{https://doi.org/10.1214/009053606000000155}{\texttt{doi:10.1214/009053606000000155}}.  
\bibitem{CorcueraGiummole1999Prediction}
Corcuera, Jose Manuel and Giummole, Federica.
\newblock A generalized Bayes rule for prediction.
\newblock \emph{Scandinavian Journal of Statistics}, 26(2):265--279, 1999.  
\newblock \href{https://doi.org/10.1111/1467-9469.00149}{\texttt{doi:10.1111/1467-9469.00149}}.  
\bibitem{Tanaka2006GeneralizedPredictive}
F. Tanaka.
\newblock Generalized Bayesian predictive density operators.
\newblock arXiv:quant-ph/0602072, 2006.  
\newblock \href{https://doi.org/10.48550/arXiv.quant-ph/0602072}{\texttt{doi:10.48550/arXiv.quant-ph/0602072}}.  
\bibitem{Holevo1979CovariantMeasurements}
Holevo, Alexander S..
\newblock Covariant Measurements and Uncertainty Relations.
\newblock \emph{Reports on Mathematical Physics}, 16(3):385--400, 1979.  
\newblock \href{https://doi.org/10.1016/0034-4877(79)90072-7}{\texttt{doi:10.1016/0034-4877(79)90072-7}}.  
\bibitem{Holevo1982ProbabilisticStatistical}
Holevo, Alexander S..
\newblock Probabilistic and Statistical Aspects of Quantum Theory.
\newblock \emph{North-Holland}, 1, 1982.  
\newblock \href{https://doi.org/10.1007/978-88-7642-378-9}{\texttt{doi:10.1007/978-88-7642-378-9}}.  
\bibitem{MassarPopescu1995OptimalEstimation}
Massar, Serge and Popescu, Sandu.
\newblock Optimal Extraction of Information from Finite Quantum Ensembles.
\newblock \emph{Physical Review Letters}, 74(8):1259--1263, 1995.  
\newblock \href{https://doi.org/10.1103/PhysRevLett.74.1259}{\texttt{doi:10.1103/PhysRevLett.74.1259}}.  
\bibitem{DerkaBuzekEkert1998OptimalEstimation}
Derka, Radoslav and Bu\v{z}ek, Vladim\'{i}r and Ekert, Artur K..
\newblock Universal Algorithm for Optimal Estimation of Quantum States from Finite Ensembles via Realizable Generalized Measurement.
\newblock \emph{Physical Review Letters}, 80(8):1571--1575, 1998.  
\newblock \href{https://doi.org/10.1103/PhysRevLett.80.1571}{\texttt{doi:10.1103/PhysRevLett.80.1571}}.  
\bibitem{Hayashi1998PureStateEstimation}
Masahito Hayashi.
\newblock Asymptotic estimation theory for a finite-dimensional pure state model.
\newblock \emph{Journal of Physics A: Mathematical and General}, 31(20):4633, 1998.  
\newblock \href{https://doi.org/10.1088/0305-4470/31/20/006}{\texttt{doi:10.1088/0305-4470/31/20/006}}.  
\bibitem{Group2}
M. Hayashi.
\newblock A Group Theoretic Approach to Quantum Information.
\newblock \emph{Springer Cham}, 2017.  
\newblock \href{https://doi.org/10.1007/978-3-319-45241-8}{\texttt{doi:10.1007/978-3-319-45241-8}}.  
\bibitem{umegaki1962}
Umegaki, Hisaharu.
\newblock Conditional expectation in an operator algebra. IV. Entropy and information.
\newblock \emph{Kodai Mathematical Seminar Reports}, 14(2):59--85, 1962.  
\newblock \href{https://doi.org/10.2996/kmj/1138844604}{\texttt{doi:10.2996/kmj/1138844604}}.  
\end{thebibliography}
\end{document}